\documentclass[twoside]{article}

\usepackage[dvips]{graphicx}
\usepackage[utf8]{inputenc}
\usepackage[english,russian]{babel}
\usepackage{amsmath,amsthm,amssymb}
\usepackage{euscript}

\RequirePackage[%
    colorlinks=true,
    pdfstartview=FitH,
    linkcolor=blue,
    citecolor=blue,
    filecolor=blue,
    urlcolor=blue,
    linktoc=page]{hyperref}

\hypersetup{%
    pdftitle={On Synthesis of Reversible Circuits with Small Number of Additional Inputs Consisting of NOT,
        CNOT and 2-CNOT Gates},
    pdfauthor={Dmitry V. Zakablukov},
    pdfkeywords={reversible logic, gate complexity, computations with memory},
    pdfsubject={The paper discusses the gate complexity of reversible circuits with the small number of additional
        inputs consisting of NOT, CNOT and 2-CNOT gates}}

\RequirePackage{array}
\RequirePackage{multirow}

\newcommand{\ZZ}{\mathbb Z}
\newcommand{\NN}{\mathbb N}
\newcommand{\frS}{\mathfrak S}
\newcommand{\vv}{\mathbf}

\theoremstyle{plain}
\newtheorem{theorem}{Теорема}

\newtheorem{corollary}{Следствие}
\newtheorem{assertion}{Утверждение}

\theoremstyle{definition}

\begin{document}

\title{О синтезе обратимых схем с малым числом дополнительных входов из элементов NOT, CNOT и 2-CNOT\\%
    On Synthesis of Reversible Circuits with Small Number of Additional Inputs Consisting of NOT, CNOT and 2-CNOT Gates}

\author{Д.\,В.~Закаблуков\footnote{Ведущий разработчик программного обеспечения, Passware, Inc., %
    e-mail: \texttt{\href{mailto:dmitriy.zakablukov@gmail.com}{dmitriy.zakablukov@gmail.com}}}\\
    Dmitry V. Zakablukov\footnote{Lead software developer, Passware, Inc., %
    e-mail: \texttt{\href{mailto:dmitriy.zakablukov@gmail.com}{dmitriy.zakablukov@gmail.com}}}}

\maketitle

УДК 004.312, 519.7    
   
Исследование выполнено при финансовой поддержке РФФИ в рамках научного проекта № 16-01-00196~A.

\begin{abstract}
В работе рассматривается вопрос сложности обратимых схем, состоящих из функциональных элементов NOT, CNOT и 2-CNOT
и имеющих малое число дополнительных входов. Изучается функции Шеннонa сложности $L(n, q)$ обратимой схемы,
реализующей отображение $f\colon \ZZ_2^n \to \ZZ_2^n$, при условии, что количество дополнительных входов $q \leqslant O(n^2)$.
Доказывается оценка $L(n,q) \asymp n2^n \mathop / \log_2 n$ для указанного диапазона значений $q$.

The paper discusses the gate complexity of reversible circuits with the small number of additional inputs
consisting of NOT, CNOT and 2-CNOT gates. We study Shannon's gate complexity function $L(n, q)$
for a reversible circuit implementing a Boolean transformation $f\colon \ZZ_2^n \to \ZZ_2^n$
with $q \leqslant O(n^2)$ additional inputs.
The general bound $L(n,q) \asymp n2^n \mathop / \log_2 n$ is proved for this case.
\end{abstract}

\textbf{Ключевые слова}: обратимые схемы, сложность схемы, глубина схемы, вычисления с памятью.


\section*{Введение}
В качестве меры сложности булевой функции можно рассматривать сложность реализующей ее минимальной схемы
из функциональных элементов или же контактной схемы, как было предложено еще К.~Шенноном в работе 1949 года~\cite{shannon}.
Схемы из классических функциональных элементов изучаются в теории сложности управляющих систем
довольно продолжительное время. Основные оценки сложности таких схем были получены О.\,Б. Лупановым уже
к 1958 году~\cite{lupanov_complexity}.

Обратимые функциональные элементы, т.\,е. элементы, реализующие биективное отображение, изучались Р.~Фейнманом и Т.~Тоффоли
в 1980-х годах в работах~\cite{feynman} и~\cite{toffoli}.
Однако, несмотря на возросший интерес к таким схемам в последнее время в связи с развитием теории квантовых вычислений,
вопрос асимптотической сложности обратимых схем подробно не изучался вплоть до последнего времени.
Первые результаты о порядке роста сложности обратимых схем, состоящих из функциональных элементов NOT, CNOT и 2-CNOT,
были получены в работах~\cite{my_first_complexity_bounds} и~\cite{my_complexity_bounds_dm}.
Уже тогда стало ясно, что сложность обратимых схем существенно зависит от количества дополнительных входов в схеме.

Для схем из классических функциональных элементов вопрос о вычислениях с ограниченной памятью был рассмотрен Н.\,А. Карповой
в работе~\cite{karpova}. Полученные оценки показали, что в базисе классических функциональных элементов,
реализующих все \mbox{$p$-местные} булевы функции, асимптотическая оценка функции Шеннона сложности схемы
с тремя и более регистрами памяти зависит от значения $p$, но не изменяется при увеличении количества
используемых регистров памяти.

В модели Карповой каждому входу и выходу вершины графа, описывающего схему, приписывается некоторый символ
из множества $R = \{\,r_1, \ldots, r_n\,\}$. Каждый символ $r_i$ можно считать уникальным идентификатором ячейки памяти,
значение из которой поступает на вход функционального элемента или в которую записывается значение с его выхода.
В модели обратимых схем из функциональных элементов NOT, CNOT и 2-CNOT, описанной в работе~\cite{my_complexity_bounds_dm},
все элементы схемы соединяются друг с другом последовательно без ветвлений и имеют одинаковое количество входов и выходов;
схему можно представить в виде нескольких параллельных линий, по которым передаются булевы значения и которые
могут быть изменены элементами схемы. Данные модели схем с ограниченной памятью весьма похожи:
регистрам памяти классической схемы соответствуют линии обратимой схемы,
поскольку и те, и другие хранят результат вычислений на каждом шаге работы схемы.
Однако если модель Карповой позволяет перезаписывать значения в ячейках памяти
(символ $r_i$, приписанный выходу некоторого элемента, может совпадать с символом,
приписанным одному из входов этого же элемента), то модель обратимой схемы не позволяет
перезаписывать значения на линиях, а лишь инвертировать их в некоторых случаях: контролируемому выходу обратимого элемента
приписывается номер линии, значение на которой будет инвертировано,
если значение булевой функции от значений на контролирующих входах элемента будет равно 1,
причем номер такой линии не может совпадать ни с одним из номеров линий, приписанных контролирующим входам элемента.

В работе~\cite{my_general_bounds} изучались функции Шеннона сложности $L(n,q)$ и глубины $D(n,q)$ обратимой схемы,
состоящей из элементов NOT, CNOT и 2-CNOT и реализующей некоторое булево отображение $\ZZ_2^n \to \ZZ_2^n$
с использованием $q$ дополнительных входов. Было доказано, что для всех значений $q$ таких,
что $n^2 \lesssim q \lesssim n2^{n - o(n)}$, порядок роста сложности обратимой схемы удовлетворяет
соотношению $L(n,q) \asymp n2^n \mathop / \log_2 q$.

В данной работе рассматриваются схемы с малым числом дополнительных входов, состоящие из обратимых функциональных
элементов NOT, CNOT и 2-CNOT. Доказывается, что для всех значений $q \leqslant O(n^2)$ верно соотношение
$L(n,q) \asymp n2^n \mathop / \log_2 n$.
Также доказывается соотношение $L(n,q) \asymp n2^n \mathop / \log_2 (n + q)$ для всех значений $q$ таких,
что $q \lesssim n2^{n-o(n)}$.

\section{Основные понятия}

Mы будем пользоваться формальным определением обратимых элементов из работы~\cite{my_complexity_bounds_dm}.
Через $C_{i_1,\ldots,i_k;j}^n$ обозначается функциональный элемент $k$-CNOT с $n$ входами
(контролируемый инвертор, обобщенный элемент Тоффоли с $k$ контролирующими входами),
задающий преобразование $\ZZ_2^n \to \ZZ_2^n$ вида
$$
    f_{i_1,\ldots,i_k;j}(\langle x_1, \ldots, x_n \rangle) =
        \langle x_1, \ldots, x_j \oplus x_{i_1} \wedge \ldots \wedge x_{i_k}, \ldots, x_n \rangle  \; .
$$
Инвертор NOT не имеет контролирующих входов, элемент CNOT имеет ровно один контролирующий вход,
а элемент Тоффоли 2-CNOT~--- ровно два.
Множество всех элементов NOT, CNOT и 2-CNOT с $n$ входами будем обозначать через $\Omega_n^2$

Сложность обратимой схемы $\frS$ будем обозначать через $L(\frS)$. Функция Шеннона сложности $L(n,q)$ обратимой схемы,
состоящей из элементов множества $\Omega_n^2$ и реализующей некоторое булево отображение $\ZZ_2^n \to \ZZ_2^n$
с использованием $q$ дополнительных входов, была определена в работе~\cite{my_complexity_bounds_dm}.

Значимыми входами схемы будем называть все входы, не являющиеся дополнительными, а значимыми выходами~--- те выходы,
значения на которых нужны для дальнейших вычислений.

В работе~\cite{shende} было доказано, что при помощи обратимых схем, состоящих только из элементов множества $\Omega_n^2$,
можно реализовать любую заданную подстановку из симметрической группы $S(\ZZ_2^n)$ при $n < 4$ и любую заданную четную
подстановку из знакопеременной группы $A(\ZZ_2^n)$ при $n \geqslant 4$. Другими словами, обратимые схемы с четырьмя и более
входами могут реализовать только те отображения, которые задают четную подстановку.

В работе~\cite{my_complexity_bounds_dm} было введено множество $F(n,q)$ всех отображений $\ZZ_2^n \to \ZZ_2^n$,
которые могут быть реализованы обратимой схемой с $(n+q)$ входами. Очевидно, что $F(n,q)$ является подмножеством
множества $P_2(n,n)$ всех булевых отображений $\ZZ_2^n \to \ZZ_2^n$. В той же работе было показано,
что $F(n,q) = P_2(n,n)$ при $q \geqslant n$ и что $F(n,0)$ совпадает с множеством отображений,
задаваемых всеми подстановками из $S(\ZZ_2^n)$ и $A(\ZZ_2^n)$ при $n < 4$ и $n \geqslant 4$, соответственно.

\section{Отображения, реализуемые обратимыми схемами с малым числом дополнительных входов}

Рассмотрим некоторое булево отображение $f\colon \ZZ_2^n \to \ZZ_2^n$.
Пусть для некоторых двух различных входов $\vv x_1$ и $\vv x_2$ значение отображения $f$ от этих входов совпадает:
$f(\vv x_1) = f(\vv x_2)$.
Определим для выхода $\vv y$ множество $A_{\vv y}$ его прообразов:
$A_{\vv y}=\{\,\vv x \in \ZZ_2^n\mid f(\vv x) = \vv y \in \ZZ_2^n\,\}$.
Обозначим через $d$ максимальное количество прообразов среди всех выходов: $d = \max\limits_{\vv y}{|A_{\vv y}|}$.

\begin{assertion}\label{assertion_max_additional_outputs}
    Не существует обратимой схемы, состоящей из элементов множества $\Omega_n^2$, реализующей заданное отображение $f$
    с $q < \lceil \log_2 d \rceil$ дополнительными входами, где $d = \max\limits_{\vv y}{|A_{\vv y}|}$.
\end{assertion}
\begin{proof}
    Докажем от противного.
    Пусть существует обратимая схема $\frS$, состоящая из элементов множества $\Omega_n^2$
    и реализующая отображение $f$ с $q < \lceil \log_2 d \rceil$ дополнительными входами.
    Поскольку $q$ является целым числом, то в этом случае $q < \log_2 d$. Следовательно, $d > 2^q$.

    Существует множество $A_{\vv y}$, мощность которого равна $d$: $|A_{\vv y}| = d$, $\vv y \in \ZZ_2^n$.
    Определим множество $A \subseteq \ZZ_2^{n+q}$ следующим образом:
    $$
        A = \{\,\vv x = \langle x_1, \ldots, x_n, 0, \ldots, 0 \rangle \in \ZZ_2^{n+q} \mid
            \langle x_1, \ldots, x_n \rangle \in A_{\vv y}\,\} \; .
    $$

    Рассмотрим булево преобразование $g\colon \ZZ_2^{n+q} \to \ZZ_2^{n+q}$,
    задаваемое схемой $\frS$.
    Для всех $\vv x \in A$ верно следующее равенство:
    \begin{align*}
        &g(\vv x) = \langle y_1, \ldots, y_n, z_1, \ldots, z_q \rangle \; , \\
        &\vv y = \langle y_1, \ldots, y_n \rangle, z_i \in \ZZ_2 \; .
    \end{align*}
    Отсюда следует, что мощность множества значений преобразования $g$ на множестве $A$
    $$
        |g[A]| \leqslant 2^q \; .
    $$
    При этом $|A| = |A_{\vv y}| = d > 2^q \Rightarrow$ на множестве $A$ преобразование $g$ сюръективно
    $\Rightarrow g$ не биективно, а значит схема $\frS$ не существует.

    Пришли к противоречию, следовательно, доказываемое утверждение верно.
\end{proof}

Для упрощения дальнейших рассуждений воспользуемся следующими отображениями, введенными в работе~\cite{my_complexity_bounds_dm}:
\begin{enumerate}

    \item
        \textit{Расширяющее} отображение $\phi_{n,n+k}\colon \mathbb Z_2^n \to \mathbb Z_2^{n+k}$ вида
        $$
            \phi_{n,n+k}( \langle x_1, \ldots, x_n \rangle ) = \langle x_1, \ldots, x_n, 0, \ldots, 0 \rangle  \; .
        $$

    \item
        \textit{Редуцирующее} отображение $\psi_{n+k,n}\colon \mathbb Z_2^{n+k} \to \mathbb Z_2^n$ вида
        $$
            \psi_{n+k,n}( \langle x_1, \ldots, x_{n+k} \rangle ) =
            \langle x_1, \ldots, x_n \rangle  \; .
        $$

\end{enumerate}

Рассмотрим множество $F(n,q)$ при $0 < q \leqslant n$. Из утверждения~\ref{assertion_max_additional_outputs} следует,
что $F(n,q)$ не включает в себя булевы отображения, для которых максимальное количество прообразов среди всех выходов
строго больше $2^q$. Покажем, что все остальные булевы отображения принадлежат $F(n,q)$.
Для этого нужно доказать, что для любого отображения $f\colon \ZZ_2^n \to \ZZ_2^n$ с максимальным количеством прообразов
среди всех выходов не более $2^q$ можно построить четную подстановку $h$,
задающую преобразование $f_h\colon \ZZ_2^{n+q} \to \ZZ_2^{n+q}$, такую, что для всех $\vv x \in \ZZ_2^n$ верно равенство
\begin{equation}
    \psi_{n+q,n}(f_h(\phi_{n,n+q}(\vv x))) = f(\vv x) \; .
        \label{formula_transformation_extension}
\end{equation}
Построение подстановки $h$ для заданного отображения $f$ будем называть дополнением отображения $f$ до четной подстановки.

Обозначим через $X \subset \ZZ_2^{n+q}$ подмножество всех векторов, у которых старшые $q$ координат равны нулю:
$X = \{\, \phi_{n,n+q}(\vv x) \mid \vv x \in \ZZ_2^n \,\}$.
Пусть для некоторого выхода $\vv y$ отображения $f$ количество прообразов больше одного: $|A_{\vv y}| > 1$.
Для всех $\vv x \in A_{\vv y}$ искомая подстановка $h$ должна давать различные значения:
$$
    f_h(\phi_{n,n+q}(\vv x_1) \neq f_h(\phi_{n,n+q}(\vv x_2),
        \;\text{где}\;\, \vv x_1, \vv x_2 \in A_{\vv y}, \; \vv x_1 \neq \vv x_2 \; .
$$
При этом мы всегда можем дополнить отображение $f$ таким образом,
что хотя бы для одного входа $\vv x \in A_{\vv y}$ значение выхода преобразования $f_h$ принадлежало множеству $X$:
$$
    \exists \vv x \in A_{\vv y} \colon f_h(\phi_{n,n+q}(\vv x)) \in X \; .
$$
Для всех остальных входов из $A_{\vv y}$ значение выходов преобразования $f_h$ не принадлежит множеству $X$.
Это следует из равенства~\eqref{formula_transformation_extension}.

Обозначим через $M_h=\{\,\vv x \in \ZZ_2^{n+q} \mid f_h(\vv x) \neq \vv x \,\}$
множество подвижных точек искомой подстановки $h$, а через $Y=\{\,\vv y = f_h(\vv x) \mid \vv x \in X \,\}$~--- множество
выходов преобразования $f_h$, первые $n$ координат которых определены заданным отображением $f$.

Выясним, чему равно максимальное значение $|M_h|$.
В худшем случае подвижны все входы $\vv x \in X$, следовательно, $|X| \leqslant 2^n$.
Очевидно, что $|Y| \leqslant |X|$, поэтому $|Y| \leqslant 2^n$.
Как было сказано выше, мы всегда можем дополнить отображение $f$ таким образом, чтобы выполнялось неравенство
$X \cap Y \neq \varnothing$. Таким образом,
\begin{equation}
    |X \cup Y| \leqslant 2^{n+1} - 1 \; .
        \label{formula_x_cup_y}
\end{equation}

Остается выяснить, сколько еще требуется подвижных точек, чтобы $f_h$ задавало четную подстановку.
Для этого необходимо определить значения $f_h(\vv x)$ для всех $\vv x \in Y \setminus X$.
Если таких значений нет (случай $Y = X$), то отображение $f$ задает подстановку на множестве $\ZZ_2^n$. В случае,
когда эта подстановка является нечетной, для ее дополнения до четной требуется ровно 2 вектора (дополнительная транспозиция)
из множества $\ZZ_2^{n+q} \setminus X$, что дает оценку $|M_h| \leqslant 2^n + 2$.

Обозначим через $f_h^{(k)}(\vv x)$, $k \in \NN_+$, следующую величину:
$$
    f_h^{(k)}(\vv x) = \overbrace{f_h(f_h(\ldots f_h}^k(\vv x) \ldots)) \; .
$$
Пусть $Y' = Y \setminus X \neq \varnothing$. Для каждого $\vv y \in Y'$ существует такой вектор $\vv x_y \in X$
и такое число $k_y \in \NN_+$, что
\begin{gather*}
    f_h^{(k_y)}(\vv x) = \vv y \; , \\
    \not \exists \vv x \in X \colon f_h(\vv x) = \vv x_y \; .
\end{gather*}
Будем называть $\vv x_y$ началом цепочки длины $k_y$ для $\vv y \in Y'$.
Таким образом, у нас определены $|Y'|$ цепочек различных длин с различными началами, являющимися векторами из $X$.
Оставшиеся вектора из $X$ входят в некоторые замкнутые цепочки (циклы),
для каждого элемента $x$ которых верно равенство $f_h^{(k_x)}(\vv x) = \vv x$, где $k_x$~--- длина соответствующей цепочки
(количество элементов в цикле).

Доопределим отображение $f$ следующим образом: пусть $f_h(\vv y) = \vv x_y$ для всех $\vv y \in Y'$,
где $\vv x_y$~---начало цепочки. Тогда отображение $f_h$ будет задавать подстановку $h'$, которая, возможно, будет нечетной.
Для того, чтобы в этом случае доопределить отображение $f$ до четной подстановки, поступим следующим образом:
\begin{itemize}
    \item если $|Y'| = 1$, $\vv y \in Y'$ (случай одной цепочки), то выберем некоторый вектор
        $\vv z \in \ZZ_2^{n+q} \setminus Y$ и положим $f_h(\vv y) = \vv z$, $f_h(\vv z) = \vv x_y$,
        где $\vv x_y$~---начало цепочки;
    \item если $|Y'| > 1$, то для некоторых $\vv y_1, \vv y_2 \in Y'$ и только для них положим
        $f_h(\vv y_1) = \vv x_{y_2}$, $f_h(\vv y_2) = \vv x_{y_1}$,
        где $\vv x_{y_1}$ и $\vv x_{y_2}$~---начало соответствующих цепочек.
\end{itemize}
Для всех остальных векторов $\vv z' \in \ZZ_2^{n+q} \setminus (X \cup Y)$,
не равных в случае одной цепочки вектору $\vv z$ (см. выше), положим $f_h(\vv z') = \vv z'$.

В итоге мы получим полностью определенное на множестве $\ZZ_2^{n+q}$ преобразование $f_h$,
задающее искомую четную подстановку $h$. При этом количество подвижных точек $M_h \leqslant 2^{n+1}$.
Это следует из неравенства~\eqref{formula_x_cup_y} и того факта, что только в случае одной цепочки мы задействуем один
дополнительный вектор (добавляем одну подвижную точку).
Следовательно, для любого значения $q > 0$ и любого отображения $f\colon \ZZ_2^n \to \ZZ_2^n$ с максимальным количеством
прообразов среди всех выходов не более $2^q$ можно построить обратимую схему,
состоящую из элементов множества $\Omega_n^2$ и реализующую отображение $f$. Таким образом, все такие отображения $f$
принадлежат множеству $F(n,q)$.

\section{Сложность обратимых схем с малым числом дополнительных входов}

Определим поведение функции $L(n,q)$, когда значение количества $q$ дополнительных входов в обратимой схеме
удовлетворяет неравенству $q \leqslant O(n^2)$.

\begin{theorem}\label{theorem_upper_bound_small_number_q}
    Для любого значения $q$ верно соотношение
    $$
        L(n,q) \lesssim \frac{192n2^n}{\log_2 n} \; .
    $$
\end{theorem}
\begin{proof}
    Любое булево отображение $f\colon \ZZ_2^n \to \ZZ_2^n$ можно реализовать обратимой схемой,
    состоящей из элементов множества $\Omega_n^2$ и имеющей не более $n$ дополнительных входов,
    поскольку $F(n,q) = P_2(n,n)$ при $q \geqslant n$.
    Следовательно, верно неравенство $L(n,q) \leqslant L(n,n)$ при $q \geqslant n$.

    В работе~\cite[Теорема 2]{my_complexity_bounds_dm} была доказана оценка сложности обратимых схем без дополнительной памяти
    \begin{equation}
        L(n,0) \lesssim \frac{48n2^n}{\log_2 n} \; .
        \label{formula_upper_bound_no_memory}
    \end{equation}
    Поскольку $L(n,q) \leqslant L(n+q, 0)$ и оценка~\eqref{formula_upper_bound_no_memory} возрастает с ростом $n$,
    то $L(n,q) \leqslant L(2n,0)$ для любого значения $q$ такого, что $q \leqslant n$,
    и $L(n,q) \leqslant L(n,n) \leqslant L(2n,0)$ при $q \geqslant n$.
    Следовательно, для всех значений $q$ верно соотношение
    \begin{equation}
        L(n,q) \leqslant L(2n,0) \lesssim \frac{96n2^{2n}}{\log_2 n} \; .
        \label{formula_upper_bound_no_memory_2n}
    \end{equation}

    В оценке~\eqref{formula_upper_bound_no_memory_2n} предполагается,
    что в худшем случае количество подвижных точек реализуемой подстановки равно $2^{2n}$
    (см. работу~\cite{my_complexity_bounds_dm}).
    Однако в предыдущем разделе мы доказали, что любое отображение $f \in F(n,q)$ можно дополнить до четной подстановки,
    имеющей не более $2^{n+1}$ подвижных точек. Отсюда следует, что для любого значения $q$ верно соотношение
    $$
        L(n,q) \lesssim \frac{192n2^n}{\log_2 n} \; .
    $$
\end{proof}

Теперь можно определить порядок роста функции $L(n,q)$ при малых значениях $q$.
\begin{theorem}\label{theorem_order_bound}
    Для любого значения $q$ такого, что $q \leqslant O(n^2)$, верно соотношение
    $$
        L(n,q) \asymp \frac{n2^n}{\log_2 n} \; .
    $$
\end{theorem}
\begin{proof}
    В работе~\cite[Теорема 1]{my_complexity_bounds_dm} была доказана общая нижняя оценка сложности обратимых схем
    $$
        L(n,q) \geqslant \frac{2^n(n-2)}{3\log_2(n+q)} - \frac{n}{3} \; .
    $$
    Для любого значения $q$ такого, что $q \leqslant O(n^2)$,
    верно соотношение $L(n,q) \gtrsim n2^n \mathop / (6 \log_2 n )$.
    Сопоставляя данную оценку и оценку из Теоремы~\ref{theorem_upper_bound_small_number_q},
    получаем порядок роста функции $L(n,q)$ из условия доказываемой теоремы.
\end{proof}

\begin{corollary}
    Для любого значения $q$ такого, что $q \lesssim 2^{n-\lceil n \mathop / \phi(n)\rceil + 1}$,
    где $\phi(n)$ и $\psi(n)$~--- любые сколь угодно медленно растущие функции такие, что
    $\phi(n) \leqslant n \mathop / (\log_2 n + \log_2 \psi(n))$,
    верно соотношение
    $$
        L(n,q) \asymp \frac{n2^n}{\log_2 (n + q)} \; .
    $$
\end{corollary}
\begin{proof}
    В работе~\cite[Утверждение 1]{my_complexity_bounds_dm} была доказана оценка
    $$
        L(n,q) \asymp \frac{n2^n}{\log_2 q}
    $$
    для всех значений $q$ таких, что $n^2 \lesssim q \lesssim 2^{n-\lceil n \mathop / \phi(n)\rceil + 1}$,
    где $\phi(n)$ и $\psi(n)$~--- любые сколь угодно медленно растущие функции такие, что
    $\phi(n) \leqslant n \mathop / (\log_2 n + \log_2 \psi(n))$.
    Сопоставляя данную оценку и оценку Теоремы~\ref{theorem_order_bound}, получаем порядок роста функции $L(n,q)$ 
    из условия доказываемого утверждения.
\end{proof}


\section*{Заключение}

В данной работе были рассмотрены обратимые схемы, состоящие из функциональных элементов NOT, CNOT и 2-CNOT
и имеющие малое число дополнительных входов $q$.

Была изучена функция Шеннона сложности $L(n, q)$ обратимой схемы с малым числом дополнительных входов,
реализующей какое-либо отображение $\ZZ_2^n \to \ZZ_2^n$. Был установлен порядок роста функции $L(n, q)$
при $q \leqslant O(n^2)$ и $q \lesssim 2^{n-o(n)}$.

В данной работе было доказано, что $L(n,q) \asymp n2^n \mathop / \log_2 (n + q)$ для любого значения $q$ такого,
что $q \lesssim 2^{n-o(n)}$.


\end{document}